\documentclass[11pt, a4paper, logo, copyright]{gensyn} 

\usepackage{xspace}             
\usepackage{multirow}            
\usepackage{rotating}            
\usepackage{tikz}                
\usepackage{float}               
\usepackage[ruled,vlined]{algorithm2e}
\usepackage{soul}                

\usepackage[colorinlistoftodos,textsize=tiny,textwidth=35pt]{todonotes}

\newcommand{\Comments}{1}

\newcommand{\mynote}[2]{\ifnum\Comments=1\textcolor{#1}{#2}\fi}
\newcommand{\mytodo}[2]{\ifnum\Comments=1%
    \todo[linecolor=#1!80!black,backgroundcolor=#1,bordercolor=#1!80!black]{#2}%
\fi}

\newcommand{\newcommenter}[3]{%
    \expandafter\newcommand\csname #1\endcsname[1]{%
        \mynote{#3!50!black}{[#2: ##1]}%
    }%
    \expandafter\newcommand\csname #1todo\endcsname[1]{%
        \mytodo{#3!20!white}{#2: ##1}%
    }%
}

\ifnum\Comments=1
    \paperwidth=\dimexpr\paperwidth+50pt\relax
    \oddsidemargin=\dimexpr\oddsidemargin+25pt\relax
    \evensidemargin=\dimexpr\evensidemargin+25pt\relax
    \marginparwidth=\dimexpr\marginparwidth+25pt\relax
\fi

\usepackage{array}
\newcolumntype{P}[1]{>{\raggedright\arraybackslash}p{#1}}
\newcolumntype{M}{>{\centering\arraybackslash\footnotesize}m{.78cm}}
\newcolumntype{S}{>{\centering\arraybackslash\tiny}m{2cm}}

\newtcbtheorem[auto counter,number within=section]{obs}%
    {Observation}{fonttitle=\bfseries\upshape, fontupper=\slshape,
        arc=0mm, colback=gensynlighterpink, colframe=gensynburgundy!60!white}%
    {theorem}

\newtcbtheorem[auto counter,number within=section]{ins}%
    {Insight}{fonttitle=\bfseries\upshape, fontupper=\slshape,
        arc=0mm, colback=gensynpink!40!white, colframe=gensynburgundy}%
    {theorem}

\usepackage{bm}

\newtheorem{theorem}{Theorem}
\newtheorem{lemma}{Lemma}
\newtheorem{proposition}{Proposition}
\newtheorem{corollary}{Corollary}
\theoremstyle{definition}
\newtheorem{definition}{Definition}

\newcommand{\bz}{\bm{z}}
\newcommand{\bs}{\bm{s}}
\newcommand{\bx}{\bm{x}}

\newcommand{\Skt}{S_{K_{T-t}}}
\newcommand{\Sktj}{S_{K_{T-t}, j}}

\newcommand{\Ex}{\mathbb{E}}

\newcommenter{gab}{GA}{magenta}
\newcommenter{saf}{SH}{teal}

\title{Evidence Markets}

\author[2]{Safwan Hossain}
\author[1]{Gabriel Andrade}
\author[1,3]{Chengqi Zang}
\author[2]{Yiling Chen}

\affil[1]{Gensyn AI}
\affil[2]{Harvard University}
\affil[3]{University of Tokyo}

\correspondingauthor{Safwan Hossain was a research intern at GensynAI during the course of this work. Email: shossain@g.harvard.edu}

\reportnumber{}     

\begin{abstract}
Modern prediction markets face two limitations that restrict their applicability in a range of settings:~(i)~they reveal what the crowd believes but not the evidence or reasoning behind those beliefs, and~(ii)~they require an event with an external ground truth that resolves at a known future date. We address these twin challenges by introducing evidence markets, a generalization of prediction markets that incentivizes the submission of evidence alongside beliefs and can be endogenously resolved using the crowd-sourced evidence if external resolution is not possible. At its core, the market uses a logarithmic market scoring rule whose liquidity parameter changes dynamically with the accumulated evidence quality. We prove that platform loss is bounded, evidence is rewarded proportional to the current market uncertainty, and can be equivalently implemented through an automated market maker. In the case where the marker resolves endogenously based on submitted evidence, we characterize how withholding evidence shifts a trader's belief about resolution and use it to prove truthful belief and evidence reporting is a always an $\varepsilon$-dominant strategy incentive compatible (DSIC) strategy. To address operational considerations, we propose evidence verification via an LLM-as-a-Judge framework with staking and give an asynchronous execution algorithm that is not bottle-necked by verification. Throughout the work, we use LLM evaluations---determining which model is best for a given task---as a salient and representative running example for our proposed market.

\end{abstract}

\begin{document}

\maketitle

\section{Introduction}
Prediction markets have experienced a meteoric rise over the past few years, moving from an academic curiosity to mainstream infrastructure for forecasting. Platforms such as Polymarket and Kalshi now host markets on elections, geopolitics, scientific milestones, and economic indicators, drawing substantial volume and attention from policymakers, journalists, and researchers. Their appeal rests on a clean theoretical foundation: under standard assumptions, the market price aggregates participants' dispersed beliefs about an uncertain event into a single, continuously updating probability estimate---one that often outperforms expert forecasts and polls. Yet despite their growing prominence, prediction markets in their current form face two fundamental limitations that constrain both the depth and the breadth of what they can offer.

The first limitation is interpretive. Prices on a prediction market summarize the weighted aggregate of participants' beliefs into a single probability, but the evidence and reasoning underlying those beliefs remain privately held by each trader. In short, the current iteration of prediction markets tell us \emph{what} the crowd believes, but not \emph{why}. Understanding this reasoning is crucial to have confidence in the forecast and justify any decisions induced by it. Decision-makers acting just on price signals are forced to infer this reasoning themselves, eroding much of the efficiency that motivated such markets in the first place. In contrast, a mechanism that incentivizes traders to disclose the evidence behind their beliefs would let the market summarize precisely this information. Future traders could then update their beliefs in light of the explanations already submitted, yielding a more informed aggregate and a richer information product overall.

The second limitation is more fundamental and pertains to which events can even be considered for prediction markets. Currently, they can only consider events which are externally resolved by time --- i.e. events whose outcome will, at some specific future date, be publicly observable and verifiable independent of the market itself. This is the natural setting for elections, sporting events, asset prices, and other phenomena that nature or institutions will resolve on a known timeline. It fails, however, for the much larger class of events whose answers depend on judgment, evidence, or interpretation rather than on the passage of time: questions about scientific replication, policy effectiveness, the quality of competing artifacts, and so on. By incentivizing evidence submission, we may be able to endogenously resolve such questions based on the submitted evidence.

The evaluation of large language models (LLMs) stands out as both a pressing and structurally illustrative example of these concerns. LLMs are released at a rapid pace that has outstripped the capacity of static benchmarks to keep up. Further, benchmarks are costly to produce, saturate quickly, leak into training data, and capture narrow slices of capability that may be relevant only to a few. It is possible to capture such context-specific evaluation by framing it as an event: \emph{Which model is better at task X?} However, there is no external, time-bound, and clear resolution of such an event; indeed resolution is itself an exercise in building an evaluation dataset (evidence aggregation). This makes LLM evaluation a natural setting for a market mechanism that crowd-sources an evidence dataset and endogenously resolves through that evidence. Ideally, participants can not only bet on which model wins, but also submit supporting evidence in the form of evaluation questions which will be used to determine the winner. 

\paragraph{Our Contribution:} We propose a novel \emph{evidence market} that addresses the two limitations above in tandem. This mechanism (1) incentivizes both belief reporting and evidence submission, and (2) can be resolved externally by realized events or endogenously through the submitted evidence. Traders are free to report either belief or evidence (or both), thereby generalizing existing prediction markets as a special case of our model. Section~\ref{section:model_evidence} outlines the formal model used in this paper and outlines our proposed market at a high-level. We observe that in endogenous resolution, traders affecting market resolution through evidence submission means there is a multitude of consistent beliefs they can hold about resolution depending on what evidence they submit. Section~\ref{sec:endogenous_resolution} formalizes this epistemic question and bounds how much selectively submitting private evidence can affect a trader's resolution belief. Section~\ref{section:payoff} then discusses the overall payoff mechanism for this market under both types of market resolution by leveraging the formalisms of logarithmic market scoring rules (LMSRs) and automated market makers (AMMs). At a high-level, we propose a dynamic liquidity parameter that \emph{decreases} as the evidence quality in the market \emph{increases} and show that this achieves $\varepsilon$-dominant strategy incentive compatibility (DSIC) for endogenous resolution and strict DSIC for exogenous resolution, with respect to both truthful belief and full evidence submission. In Section~\ref{section:practical} we comment on the evidence verification problem and highlight the unsuitability of standard approaches like peer-prediction for this problem. Building on past related works, we propose an LLM-as-a-Judge framework and augment it with staking where relevant to address this challenge. We also give a practical algorithm for executing orders that is not blocked by evidence verification. We conclude in Section~\ref{sec:discussion} with a discussion on the broader implications of our work and open questions it raises.

\section{Related Work}\label{sec:related_works}

\paragraph{Prediction markets and market scoring rules:} A long line of empirical work has established that prediction markets aggregate dispersed information effectively and often outperform alternative forecasting methods~\citep{wolfers2004prediction, manski2006interpreting, arrow2008promise}. The theoretical foundation for our mechanism rests on the seminal work of Hanson which proposed logarithmic market scoring rule, an incentive-compatible mechanism for truthful belief submission~\citep{hanson2003combinatorial, hanson2007logarithmic}. Building on this, \citet{chen2007utility} formalizes the connection between such market scoring rules being duals of automated market makers (AMMs) with a given cost function, allowing versatility in implementation. Our construction inherits the LMSR cost-function structure but modifies the liquidity parameter to be a decreasing function of cumulative evidence quality, coupling price-formation directly to the evidence record.

\paragraph{Markets for unverifiable outcomes and explanations:} The self-resolving prediction market of ~\citet{srinivasan2023self} also looks to address one of the limitations we point out: existing markets require an externally verifiable ground truth to resolve. Their mechanism terminates randomly with probability $\alpha$ after each report and pays earlier participants based on the final reporter's prediction via negative cross-entropy, achieving truthful reporting as a perfect Bayesian equilibrium without ever observing an outcome. Our work shares the motivation but takes a structurally different approach. Rather than treating the final belief report as a proxy for ground truth, we open up a second axis, submission of evidence that can discriminate between alternatives, and resolve endogenously using the crowd-sourced evidence. This shift converts the resolution signal from a belief aggregate into an evidence aggregate, and thereby grounds resolution away from self-referential beliefs and toward externally inspectable artifacts that can be of independent interest. 

The other limitation mentioned -- that markets reveal beliefs but not the reasoning behind them -- is taken up by \citet{srinivasan2025tellmewhy}, which proposes directly and only eliciting rationales for events as opposed to beliefs. Their model captures the insight that traders' beliefs are drawn from overlapping information sources, so belief reports alone fail to reveal what is shared versus what is new; rationales uncover this distinction and enable more efficient aggregation. Our mechanism pursues the same goal -- uncovering the \emph{why} behind the price -- but through a different channel: rather than only eliciting verbal rationales, we make evidence/explanations a market commodity that sits alongside beliefs. Importantly, this allows traders who may not be able to articulate their rationales well or have no novel information, to still trade in the market through just beliefs as in standard prediction markets. Indeed, a trader's information/explanation may already exist in the market, but they can draw different conclusions from it and thus arrive at a distinct belief. Our work allows this flexibility not found in \citet{srinivasan2025tellmewhy}. 

On the whole, these two papers independently address the resolution problem and the explanation problem that motivated this work. To our knowledge, no prior work addresses both simultaneously. This is our core contribution: a single mechanism that can reveal the motivating basis of the price as well as endogenously resolve the market through such evidence. Crucially, our design is also \emph{flexible} in a way that neither prior mechanism is; traders may submit beliefs without evidence (recovering standard prediction markets trading), submit evidence without taking a position (acting as pure contributors to the resolution record), or do both. This flexibility lowers the participation barrier for traders whose comparative advantage lies on only one side or have different risk profiles. 

\paragraph{Peer prediction.} A separate literature elicits truthful reports without ground truth by paying agents based on the reports of peers~\citep{miller2005eliciting, prelec2004bayesian, dasgupta2013crowdsourced, kong2019information, schoenebeck2020learning}. Multi-task extensions including Correlated Agreement~\citep{shnayder2016informed} and Determinant-based Mutual Information~\citep{kong2020dominantly} strengthen these results under weaker assumptions. We engage directly with this literature as a possible way to verify evidence in the absence of ground truth. Unfortunately, when agents can participate as both traders and peer verifiers of evidence, peer predictions guarantees do not map cleanly. This motivates our optimistic verification with staked disputes, in which a frontier LLM serves as default verifier and agents dispute by staking, with adjudication having access to the whole history. This draws on the now well-established asymmetry between generation and verification of reasoning: process reward models~\citep{lightman2023verify}, generative reward models~\citep{zhang2024genrm}, and verifier-augmented thinking models~\citep{khalifa2025thinkprm} all exploit the fact that checking a candidate solution is structurally easier than producing one. Lastly, the recent WOMAC mechanism~\citep{srinivasan2025womac} also explores adjudication structures for prediction competitions and is conceptually adjacent, though aimed at a different (competition-style) setting.

\paragraph{LLM evaluation.} The most pressing application for our framework is the comparative evaluation of LLMs, where static benchmarks saturate quickly, suffer from contamination, and capture narrow capability slices~\citep{liang2022helm, srivastava2023beyond}. Crowd-sourced pairwise-preference platforms such as Chatbot Arena~\citep{chiang2024chatbot} and judged benchmarks such as MT-Bench~\citep{zheng2023judging} partially address these issues but lack incentive structure for participants and do not produce a market-style continuously-updated estimate. Our mechanism is designed for precisely this regime: the evaluation question/event can be highly targeted and the crowd-sourced evidence set can be used for resolution, alongside more detailed model evaluation and other downstream tasks.

\section{Model}\label{section:model_evidence}
\paragraph{Agent Evidence and Beliefs: }Consider an event with $n$ possible alternatives or outcomes. Participants, agents, or \emph{traders} arrive to the market sequentially, and each trader $t$'s belief about the resolution of this event is denoted by a vector $q_t \in \Delta(n-1)$, the $n-1$ dimensional simplex. Each trader also holds some evidence set $E_t = (e_{t1}, \dots, e_{t\ell})$ where each $e_{tj}$ can be see as some atomic unit of evidence/explanations\footnote{We consider traders as innately possessing their evidence pair. While another natural model could consider some effort/cost in evidence acquisition, this is not our setting.}. In the specific case of LLM evaluation, this atomic evidence can be seen as an evaluation question and answer pair: $e_{tj} = (\text{question}_{tj}, \text{answer}_{tj})$. In general, we consider traders as being able to selectively choose or subsample the evidence they submit to the market, denoted $E'_t$. Note that this does not include them manipulating or tampering the atomic units $e_{tj}$. Section~\ref{section:practical} broadens this perspective and considers the evidence verification problem for (1) exogenous resolution and (2) endogenous resolution for LLM evaluations. For endogenously resolving based on evidence, we make the following assumption to connect evidence to outcomes:

\begin{definition}
    Each alternative or outcome $i \in [n]$ has a binary outcome $X_{ie}$ with respect to an atomic evidence $e$. $X_{ie} = 1$ denotes the evidence supporting or corroborating event alternative $i$, and $X_{ie} = 0$ denotes the evidence refuting this alternative.
\end{definition}

While the definition above is simplistic, it allows theoretical tractability and is especially pertinent in the LLM evaluation setting where $e$ is a question, answer pair; $X_{ie}$ simply maps to whether model $i$ answered that question correctly or not. Lastly, note that the total evidence set a trader $t$ holds can be empty --- $E_t = \emptyset$ --- which allows us to capture existing prediction markets as a special case of our model. We use $E_{t}^{total}$ to denote the total amount of evidence in the market at some time $t$, and this is publicly available. Key to our model, is the notion of \emph{evidence quality} which is applied to arbitrary sets of evidence. More formally:

\begin{definition}~\label{def:evidence_quality}
    For a set of evidence $E = (e_1, \dots, e_\ell)$, let $r(E)$ denote the quality of that evidence set. Further, let $R_t = r(E_t^{total})$ denote the quality of the cumulative evidence set that exists in the market at time $t$. We assume that $r(\cdot)$ is non-negative and monotone increasing, i.e.~$r(E) \geq 0$ and for $E_1 \subseteq E_2$, $r(E_1) \leq r(E_2)$
\end{definition}
The quality function $r$ is clearly contextual and our results do not make any assumptions beyond the stated non-negativity and monotonicity, which are fairly natural for such settings. The function $r$ also encapsulates any verification or filtering that is pertinent for a given market and could depend on $E_t^{total}$: we may, for instance, only want relevant questions that are not semantically similar to evidence that is already in the market. We comment more on the specific instantiations of $r$ and it's nuances as it relates to both endogenous and exogenous resolution in Section~\ref{section:practical}.

\paragraph{Market Desiderata: }Our goal is to propose a market for the submission of beliefs and/or evidence which can aggregate these diverse positions and payout traders in a way that incentivizes truthful behaviour. More specifically, we wish to build a market that satisfies the following desiderata:

\begin{definition}\label{def:axioms}
    Evidence markets should allow traders to submit beliefs and/or evidence and satisfy:
    \begin{itemize}[itemsep=0.0ex,leftmargin=4ex]
         \item Axiom 1: Can be resolved exogenously or endogenously through the collected evidence.
        \item Axiom 2: It is optimal for each trader $t$ to submit their full evidence set $E_{t}$.
        \item Axiom 3: It is optimal for each trader $t$ to report/act with respect to their true belief $q_t$.
        \item Axiom 4: The numerical market belief at any time should be clearly interpretable.
        \item Axiom 5: The order in which an agent submits her evidence/beliefs/trades should be irrelevant. 
\end{itemize}
\end{definition}

\begin{figure}
    \centering
    \includegraphics[width=0.8\linewidth]{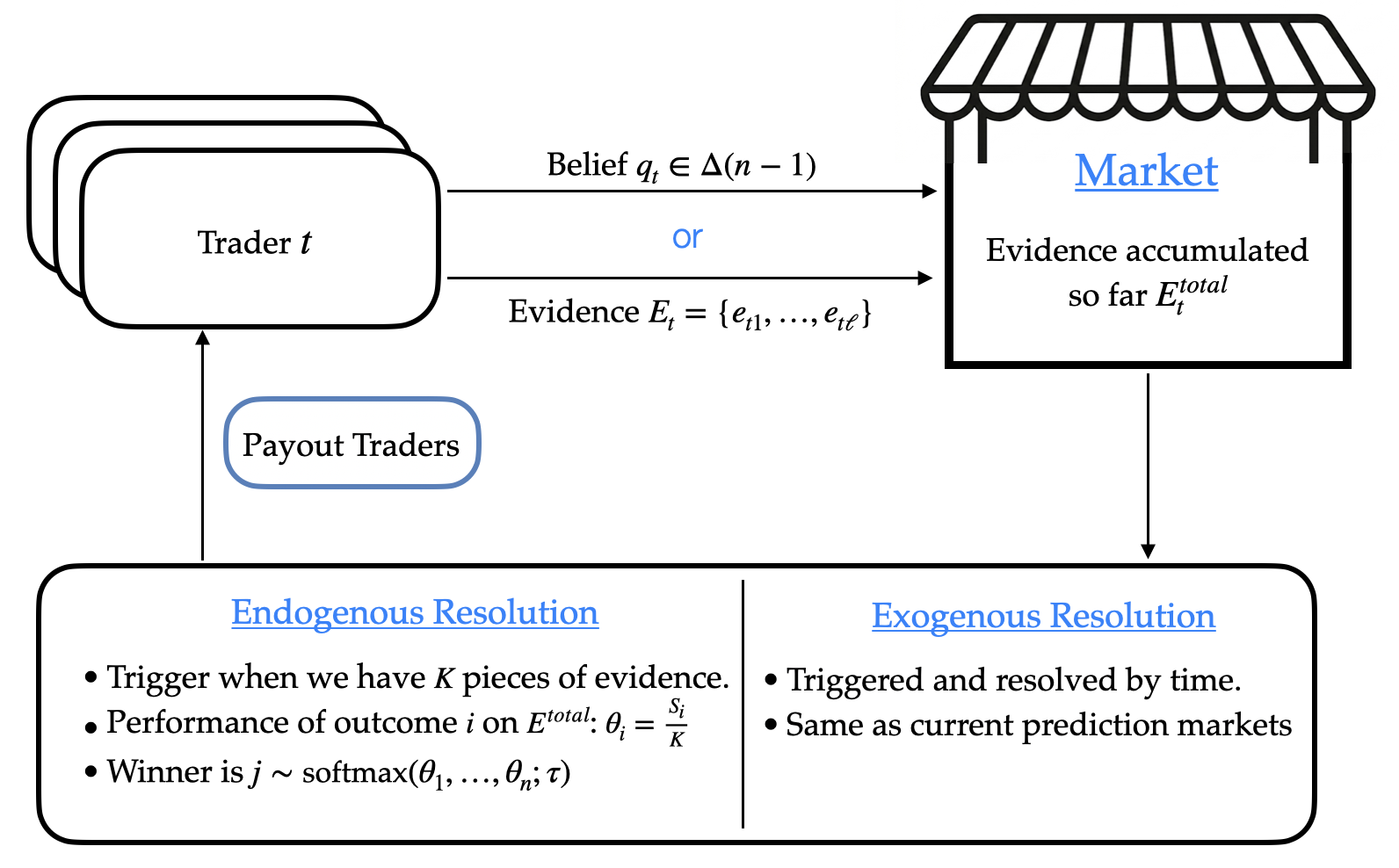}
    \caption{High-Level flow of the proposed evidence market. Section~\ref{sec:endogenous_resolution} studies the connection between belief and evidence under endogenous resolution in detail, and Section~\ref{section:payoff} outlines our specific payout mechanism that build upon and generalizes LMSRs.}
    \label{fig:market_flow}
\end{figure}

\paragraph{Market Structure:} We now propose our high-level market structure, visualized in Figure~\ref{fig:market_flow}. It consists of traders arriving sequentially and submitting beliefs and evidence to the market. Beliefs can be submitted directly or indirectly through buying shares of resolution (see Section~\ref{sec:amm_payoff}). If the market is externally resolved, the resolution is triggered at the associated time $T$, and traders receive payoff. In the case of endogenous resolution, we propose resolution be triggered when $K$ pieces of evidence has accumulated, with $K$ being a event-creator chosen hyperparameter. Resolution then follows as such:
\begin{definition}\label{defn:market_resolution}
    Endogenous resolution is triggered when $K = |E^{total}|$ pieces of evidence has accumulated. For each resolution alternative $i$, let $\theta_{i} = \tfrac{1}{K}\sum_{e} X_{ie}$ denote the fraction of $E^{total}$ supporting $i$. Then the market resolves by sampling alternative $i \sim \text{softmax}(\theta_{1}, \dots, \theta_{n}; \tau)$ where $\tau$ is the temperature of the softmax.
\end{definition}

The remainder of the paper is devoted to precisely specifying each part of this market and validating that they satisfy the market desiderata outlined. We first consider the epistemic question that arises in endogenous resolution: how does selective evidence submission affect a trader's belief about resolution and what is the sensitivity of this relationship. Later, we outline the market's payoff mechanism, a modified version of an LMSR that can be equivalently implemented through an automated market maker (AMM).

\section{Beliefs under Endogenous Resolution}\label{sec:endogenous_resolution}
In traditional prediction markets, individuals cannot affect outcome resolution. As such, their belief $q$ about market resolution is an exogenous quantity and coincides with their belief about how the outcome will be realized. This structurally changes in endogenous resolution where evidence submitted by the traders is used to resolve the market. As such, a distinction emerges here between a trader's belief about the underlying event, and their belief about market resolution, which in our proposal is sampled from the softmax over the terminal scores of each alternative. It is a trader's belief about this latter that we denote by $q$. To be precise about this, first observe that we can decompose the score of alternative $j$ from the perspective of trader $t$ as a function of the current evidence in the market, the evidence submitted by trader $t$, and future evidence. Specifically, let $k(t) = |E^t_{total}|$ denote the total pieces of evidence collected at time $t$, with $k_j(t)= \sum^{k(t)}X_{je}$ counting how much of the current evidence supports alternative $j$. For trader $t$ holding evidence set $E_t$, let $E'_t \subseteq E_t$ of which $E'_{t,j}$ support outcome $j$. Lastly, let $\Skt = K - k(t) -|E'_t|$ denote the amount of future evidence left to arrive, with the random variable $\Sktj$ denoting the number of future evidence that supports alternative $j$. Then for any $E'_t$ this trader may submit:
\begin{equation*}
    \text{Trader $t$'s view of alternative $j$'s score: } \theta_{tj}(E'_t) = \frac{k_j(t) + |E'_{t,j}| + \Sktj}{k(t) + |E'_{t}| + \Skt} = \frac{k_j(t) + |E'_{t,j}| + \Sktj}{K}.
\end{equation*}
Clearly, trader $t$'s perception of the score for alternative $j$ is a function of the evidence they choose to submit, and correspondingly, so is their belief about resolution\footnote{We generally only consider the trader as deciding to strategically choose a subset of their entire evidence to submit. Manipulation in the form of fabricating or tampering evidence is part of verification which we discussed in Section~\ref{sec:verification}.}. This is formally given below, with $\bm{S_{K_{T-t, *}}} = (S_{K_{T-t, 1}}, \dots, S_{K_{T-t, n}})$:
\begin{lemma}\label{lemma:belief_endogenous}
    A trader arriving at time $t$ has the following belief about the market resolving to outcome $j \in [n]$ when they submit an evidence set $E'_t \subseteq E_t$:
    \begin{equation}
          q_{tj}(E'_t) = \Ex_{\bm{S_{K_{T-t, *}}}}\left[\frac{\exp(\tfrac{1}{\tau}\theta_{tj}(E'_t))}{\sum_{\ell}\exp(\tfrac{1}{\tau}\theta_{t\ell}(E'_t))}\right]
    \end{equation}
\end{lemma}
\begin{proof}
    Due to the randomness of future evidence, $\theta_{tj}$ is a random variable. As per our resolution procedure (see Definition~\ref{defn:market_resolution}), we take the softmax based on these scores. So the interim belief of trader $t$ regarding resolution can be expressed as a function of the evidence they submit:
    \begin{align*}
        q_{tj}(E'_t) &= P[j \sim \text{softmax}(\theta_{t1}(E'_t), \dots, \theta_{tn}(E'_t); \tau)]\\
        &= \sum_{\bm{S_{K_{T-t, *}}}} P(\bm{S_{K_{T-t, *}}}) \, \text{softmax}(\theta_{t1}(E'_t), \dots, \theta_{tn}(E'_t); \tau |\bm{S_{K_{T-t, *}}})\\
        &= \sum_{\bm{S_{K_{T-t, *}}}} P(\bm{S_{K_{T-t, *}}}) \frac{\exp(\tfrac{1}{\tau}\theta_{tj}(E'_t))}{\sum_{\ell}\exp(\tfrac{1}{\tau}\theta_{t\ell}(E'_t))}
        = \Ex_{\bm{S_{K_{T-t, *}}}}\left[\frac{\exp(\tfrac{1}{\tau}\theta_{tj}(E'_t))}{\sum_{\ell}\exp(\tfrac{1}{\tau}\theta_{t,\ell}(E'_t))}\right]\\
    \end{align*}
\end{proof}
The sensitivity of a trader's belief to the evidence they submit is not, in general, a desirable feature. This can lead to traders strategically sub-sampling their evidence to change their corresponding belief to ensure higher payoffs. This in turn means the evidence collected by the market cannot claim to be a representative of the crowd's wisdom, diminishing its usefulness. As such, our goal as the market designer is to ensure that traders can only change their belief by some arbitrarily small $\varepsilon$ through this selective submission. If combined with a trader payoff that is (1) bounded and (2) incentivizes truthful belief reporting (see Section~\ref{section:payoff}), we can claim $\varepsilon$ incentive compatibility to submit the entire evidence set along with the corresponding true belief. We now prove the first part of this and show the sensitivity of endogenous resolution to evidence submission can be made arbitrarily small.

\begin{theorem}\label{theorem:belief_sensitivity}
    Consider a trader who arrives at time $t$ with total evidence $E_t$ which, if submitted fully, would yield a belief $q(E_t)$. Then by withholding evidence and submitting a subset $E'_t \subseteq E_t$, their corresponding belief shifts $||q(E_t) - q(E'_t)||_1 \leq \frac{1}{\tau} \frac{|E_t|}{K}$. Correspondingly, for any $\varepsilon > 0$, setting $\tau \geq \frac{|E_t|}{K\varepsilon}$ ensures that $||q(E_t) - q(E'_t)||_1 \leq \varepsilon$.
\end{theorem}
\begin{proof}
    From Lemma~\ref{lemma:belief_endogenous}, we note that this difference in belief can be expressed as follows:
    \begin{align*}
        ||q(E_t) - q(E'_t)||_1 &= \sum_{j=1}^{n}{\bigg| \sum_{\bm{S_{K_{T-t, *}}}} P(\bm{S_{K_{T-t, *}}}) \frac{\exp(\tfrac{1}{\tau}\theta_{tj}(E_t))}{\sum_{\ell}\exp(\tfrac{1}{\tau}\theta_{t,\ell}(E_t))} - \frac{\exp(\tfrac{1}{\tau}\theta_{tj}(E'_t))}{\sum_{\ell}\exp(\tfrac{1}{\tau}\theta_{t,\ell}(E_t'))}\bigg|} \\
        &=  \sum_{\bm{S_{K_{T-t, *}}}} P(\bm{S_{K_{T-t, *}}}) \sum_{j=1}^{n}\bigg| \text{softmax}(\theta_{t1}, \dots, \theta_{tn} | \bm{S_{K_{T-t, *}}}, E_t) - \text{softmax}(\theta'_{t1}, \dots, \theta'_{tn} |\bm{S_{K_{T-t, *}}}, E_t')\bigg|
    \end{align*}

For brevity, let $\bx = (\theta_{t1}, \dots, \theta_{tn})$ under a realization of $\Sktj$ and submission of true evidence $E_t$. Further, let $\sigma(\cdot) := \text{softmax}(\cdot)$. We will first determine the sensitivity of the softmax function at $\bx$ for any perturbation $\delta$. We first note the Jacobian of the softmax is the following:
\begin{equation}
    \frac{\partial \sigma_j}{\partial \bx_{m}} = \frac{1}{\tau}\sigma_j(\bx)\bigg( \bm{1}[j = m] - \sigma_m(\bx)\bigg)
\end{equation}
For any direction $\delta$, we then observe the following using the identities $(1 - \sigma_j) = \sum_{m \ne j}{\sigma_m}$ and $\sum_{j}\sigma_j = 1$:
\begin{align*}
    \nabla \sigma_j(\bx) \cdot \delta &= \frac{1}{\tau}\sum_{m=1}^n\sigma_j(\bx)\bigg( \bm{1}[j = m] - \sigma_m(\bx)\bigg) = \frac{1}{\tau}\delta_j \sigma_j \sum_{j \ne m}{\sigma_m} - \frac{1}{\tau}\sum_{m \ne j}{\delta_m \sigma_j \sigma_m} \\
    &= \frac{1}{\tau}\sigma_j\left(\sum_{m \ne j}{\delta_j \sigma_m} - \delta_m \sigma_m\right) = \frac{1}{\tau} \sigma_j\left( \delta_j(1 - \sigma_j) + \sigma_j \delta_j - \sigma_j \delta_j - \sum_{m \ne j}{\sigma_m \delta_m} \right) \\
    &= \frac{1}{\tau} \sigma_j\left(\delta_j - \sum_{m}{\sigma_m \delta_m} \right) := \frac{1}{\tau} \sigma_j\left(\delta_j - \bar{\delta} \right)
\end{align*}
where $\bar{\delta} =  \sum_{m}{\sigma_m \delta_m}$ is simply the weighted average of the $\delta$ vector according to the weights of $\sigma_m$. By the mean value theorem, there exists a $u \in [0,1]$ such that: $\sigma_j(\bx + \delta) - \sigma(\bx) = \nabla \sigma_j(\bx + u\delta) \cdot \delta$. Then using the above, we have that:
\begin{align}\label{eq:lipschitz_bound}
    \sum_{j=1}^{n}|\sigma_j(\bx + \delta) - \sigma(\bx)| = \frac{1}{\tau}\sum_{j=1}^{n}{|\sigma_j(\bx + u \delta)(\delta_j - \bar{\delta})}|
\end{align}
Let $J^+$ denote the set of indices where $\delta_j - \bar{\delta} \geq 0$ and let $J^-$ denote the set of indices where $\bar{\delta} - \delta_j \geq 0$. Note that $\sum_{j=1}^{n}{\sigma_j(\delta_j - \bar{\delta})} = 0$ by definition: softmax sums to 1 and $\bar{\delta}$ is the weighted average with respect to $\sigma_j(\bx + u \delta)$. As such $V = \sum_{j \in J^+}{\sigma_j(\delta_j - \bar{\delta})} = \sum_{j \in J^-}{|\sigma_j(\delta_j - \bar{\delta})|}$. So it suffices to bound $V$. Let $p = \sum_{j \in J^+}{\sigma_j(\bx + u \delta)}$. Then the following holds:
\begin{gather*}
    V \leq p \cdot \max_{j \in J^+}{\delta_j - \bar{\delta}} \quad \text{and} \quad V \leq (1-p) \cdot \min_{j \in J^-}{\bar{\delta}- \delta_j} \\
    \implies V^2 \leq p(1-p)\frac{\Delta \delta^2}{4} \leq \frac{\Delta \delta^2}{16} \implies V \leq \frac{\Delta \delta}{4}
\end{gather*}
where we use the fact that $p(1-p) \leq \tfrac{1}{4}$ and letting $\Delta \delta = \max_{ij}{\delta_i - \delta_j}$, we use the AM-GM inequality since $\max_{j \in J^+}{\delta_j - \bar{\delta}} + \min_{j \in J^-}{\bar{\delta} - \delta_j} = \Delta \delta$. Thus we can claim that:
\begin{equation}
    \sum_{j=1}^{n}|\sigma_j(\bx + \delta) - \sigma(\bx)| \leq \frac{1}{\tau}2V \leq \frac{\Delta \delta}{2 \tau}
\end{equation}

A trader manipulating essentially means selectively choosing which evidence to share. Let $M = |E_t|$, $M' = |E'_t|$. Being strategic is equivalent then to withholding $W = M - M'$ of evidence. Let each element of $\bx$ be the $\theta^T_{tj}$ under full revelation of $M$. Since the number of future points is fixed at $\Skt$ for a given $M$ number of points submitted, let $S_j = k_j(t) + M_{j} + \Sktj$ denote the number of those total evidence points that support alternative $j$. Then $\bx_j = \theta_{tj} = \frac{S_j}{K}$. 

Next, suppose that this trader manipulated by withholding $W$ pieces of evidence. Let $d_j$ denote the number of supports alternative $j$ lost due to this change vis-a-vis the trader $t$ submitting their true set $M$. The deviating $\tilde{\theta}_{tj}$ that will be used to compute the softmax under this manipulated evidence submission is $\tilde{\theta}_{tj} = \frac{S_j - d_j + Z_j}{K}$, where $Z_j$ denotes the realization (for alternative $j$) of the new evidence that will fill in for the withheld $W$. Comparing against $\theta_{tj}$ (or $x_j$), the corresponding $\delta_j$ such that $\tilde{\theta}_{tj} = \theta_{tj} - \delta_j$ is:
\begin{equation}
    \delta_j  = \frac{S_j}{K} - \frac{S_j - d_j + Z_j}{K} = \frac{d_j - Z_j}{K}
\end{equation}
Next, for any two parameters $\delta_j$ and $\delta_m$, we have that:
\begin{equation}
    |\delta_j - \delta_m| \leq \frac{|d_j - d_m| + |Z_j - Z_m|}{K} \leq \frac{2W}{K} \leq \frac{2M}{K} := 2w
\end{equation}
where $w = \frac{M}{K}$ denotes what the evidence this trader holds as a fraction of the total evidence needed to resolve the market. Plugging this into the Lipschitz bound derived in equation $\ref{eq:lipschitz_bound}$, we have that the maximum change in beliefs due to withholding evidence is $\frac{w}{\tau}$.
\begin{align}
    |q(M) - q(M')|_1 \leq \sum_{\Sktj = 0}^{\Skt} P(\Sktj) \frac{w}{\tau} = \frac{w}{\tau}
\end{align}
If we wish to ensure there is no more than $\epsilon$ difference between beliefs due to evidence submission, then we can simply choose the temperature $\tau$ as follows:
\begin{equation}
    \Delta q \leq \frac{M}{\tau K} \leq \varepsilon \implies \tau \geq \frac{M}{\varepsilon K} = \frac{w}{\varepsilon}
\end{equation}
\end{proof}

We make a few observations. First the limit to which any trader can affect their belief about resolution depends on the relative size of their evidence with respect to $K$. Crucially, it does not depend on what that future evidence looks like (i.e. which alternatives they support). For the platform then, it suffices to estimate the largest evidence whale in the market. The larger this whale, the larger the temperature $\tau$ becomes to diminish how much they can change their belief through selective evidence submission. 

While this approach can ensure arbitrarily small deviations, it means that market resolution becomes less sharp, especially in the presence of whales. Contrasting to exogenously resolved markets where the market belief gets sharper as one gets closer to resolution, this may not be the case here. This is less of a dilemma than it appears for two reasons. Firstly, resolution time is binding in the exogenous case and as such, it is natural that as one gets closer to it the more uncertainty disappears and the sharper the beliefs. In the endogenous case, the choice of resolution trigger $K$ is arbitrary. Secondly, in exogenous resolution, trader beliefs reflect their opinion on the unrealized world event and are of primary importance. By contrast, in endogenous markets, the submitted beliefs are about resolution and can be viewed separately from the metric implied by the actual evidence being submitted. The event-creator may, for instance, care about the argmax of the scores rather than the softmax lottery. So long as traders are incentivized to submit their full evidence set, any desired function of the cumulative evidence can be faithfully computed, independent of the resolution mechanism's sharpness.

\section{Payoff Mechanism}\label{section:payoff}
\subsection{Direct Belief Elicitation and Market Scoring Rules}
\paragraph{Classical Perspectives:} Having explored the relationship between market beliefs in endogenous resolution and evidence, we now turn to the question of designing payoffs for our proposed market. This is of central importance in both exogenous and endogenous resolution since payoff structure determines whether traders are correctly incentivized to submit true beliefs and evidence. Our starting point is market scoring rules: a family of incentive-compatible mechanisms that aggregate probabilistic forecasts from multiple participants. They are the workhorse of academic literature on prediction markets, with the most canonical version being the logarithmic market scoring rule. Traders arrive sequentially and submit their complete beliefs over the space of alternatives. At event realization, their payoff is proportional to how well they can predict the realized outcome as compared to the trader immediately before them. Specifically, the payoff considers the difference between the log probabilities assigned to the realized outcome, scaled by a \emph{liquidity parameter} $\beta$. The formal definition is given below:

\begin{definition}[Logarithmic Market Scoring Rule (LMSR)]
    Consider traders arriving sequentially with trader $t$ holding a private belief $q_t$. They may submit a possibly different belief $\hat{q}_t$ to the mechanism and their payoff when the event resolves to outcome $\omega$ is given by:
    \begin{equation*}
        \text{payoff}_t(\hat{q}_{t}) = \beta \log \hat{q}_{t}^{\omega} - \beta \log q_{t-1}^{\omega} 
    \end{equation*}
    In such markets, it is (1) dominant strategy incentive compatible (DSIC) for all traders to submit their true belief $q'=q$ and (2) the platform's loss is bounded by $\beta \log(n)$.
\end{definition}

There are two core properties of this mechanism. First, it is strictly optimal for each trader to submit their true belief $q_t$ and not manipulate it. Second, while the cumulative amount (sum of all the payoffs) paid out by the platform running this mechanism is non-zero -- i.e. the platform can lose money -- this is bounded by $\beta \log n$. To see this, observe that sum of the payoffs form a telescoping series, with the total payout being the difference between the log probabilities assigned in the first ($t=0$) and last belief ($t=T$). In the worst case, the market starts out at uniform belief and converges to the resolved winner with probability 1. While the platform losing money can seem unsatisfactory, this loss can be economically interpreted as the cost of information aggregation with strong incentive guarantees.

Note that trader $t$ can also lose money in the mechanism (i.e. have negative payoff) if they were worse at predicting than their predecessor. This is natural in any market with risk. Lastly and from an operational perspective, observe that traders submit complete beliefs up-front and receive or pay out money at market resolution.

\paragraph{Evidence-Augmented LMSR:} We now modify this mechanism to our generalized prediction market that allows for the submission of evidence alongside beliefs. Recall that $r(E)$ is a metric to compute the quality of an evidence set $E$, and $R_t = r(E_t^{total})$ denotes the quality of the cumulative evidence in the market at time $t$. The core novelty we introduce is a \emph{dynamic} liquidity parameter $\beta(\cdot)$ that adjusts based on cumulative evidence quality $R_t$. We formalize this below:

\begin{definition}[Evidence-Augmented LMSR]
    Consider a trader arriving at time $t$ with evidence $E_t$. Then their payoff in submitting a subset of evidence $E'_t \subseteq E_t$ and a belief $\hat{q}_t$ (which can be distinct from their true belief) when the resolution outcome is $\omega$ is given by\footnote{For exogenously resolved events, traders come to market with some true belief $q_t$ that is unaffected by what evidence they choose to submit. For endogenous resolution, the relationship between evidence submission and corresponding true belief is denoted $q_t(E'_t)$ and is outlined formally  in~Lemma~\ref{lemma:belief_endogenous}}: 
    \begin{align*}
        \text{payoff}_t(\hat{q}_t, E'_t) &= \beta(r(E_{t-1}^{total} \cup E'_t)) \log \hat{q}_t^{\omega} - \beta(r(E_{t-1}^{total})) \log q_{t-1}^{\omega} \\
        &= \beta(R_t)\log \hat{q}_t^{\omega} - \beta(R_{t-1}) \log q_{t-1}^{\omega}
    \end{align*}
\end{definition}

Observe that if a trader $t$ provides evidence that increases the cumulative quality, their log probability will be scaled by a different $\beta$ than the trader before them at $t-1$. In providing no evidence (or evidence of 0 marginal quality according to $r$) they face the same liquidity $\beta$ as their predecessor and their payoff is similar to classical LMSR payoff. In other words, this mechanism is a strict generalization. Observe that since the trader is free to submit any belief and evidence subset to the mechanism and the platform wants to limit exposure, our goal is to ensure (1) platform loss is bounded, (2) traders are incentivized to report true belief, which can correspond to the evidence they submit for endogenous resolution, and (3) traders are incentivized to submit their whole evidence set. The first two goals mirror the guarantees provided by standard LMSRs and we start with this.

\begin{proposition}\label{prop:lmsr_belief_ic}
    Suppose trader $t$ submits evidence set $E'_t \subseteq E_t$ and holds true belief $q_t$ or $q(E'_t)$ for exogenous or endogenous resolution respectively. Then under the Evidence Augmented LMSR, it is optimal for them to submit $q_t$ or $q(E'_t)$ in either respective case. Further, the platform loss is bounded by $\beta(R_0)\log n := \beta_0 \log n$.
\end{proposition}
\begin{proof}
    Fix the evidence submission $E'_t$ and for brevity, let $q_t$ denote the true belief (so $q_t = q_t(E'_t)$ for endogenous resolution). Let $\beta_t = \beta(R_t)$ and $\beta_{t-1} = \beta(R_{t-1})$. Next, observe that trader $t$'s objective can be simplified as such:
    \begin{equation}
        \max_{\hat{q}} \Ex_{q_t}\left[ \beta_t \log \hat{q}_t - \beta_{t-1} \log q_{t-1}\right] = \max_{\hat{q}} \beta_t \cdot \Ex_{q_t}[\log \hat{q}]
    \end{equation}
    Next, observe the following where the first inequality holds due to Jensen's
    \begin{equation}
        0 = \log \Ex_{q_t}{\left[ \frac{\hat{q}_t}{q_t}\right]} \geq \Ex_{q_t} \left[ \log \hat{q}_t \right] - \Ex_{q_t} \left[ \log q_t \right]
    \end{equation}
    As such, the trader's payoff is maximized when they report their true belief $q_t$. As for the platform loss, observe that this reduces to a telescoping sum:
    \begin{align*}
        \sum_{t=0}^{T}{\beta(R_t)\log q^{\omega}_t - \beta(R_{t-1})\log q^{\omega}_{t-1}} = \beta(R_T)\log q^{\omega}_T - \beta(R_0)\log q^{\omega}_0
    \end{align*}
    Since log probabilities are negative, this is maximized when final belief is 1 on the resolved outcome and the initial belief is uniform. This yields a maximum loss of $\beta(R_0)\log n$.
\end{proof}

\begin{figure}[htbp]
    \centering
    \begin{minipage}[c]{0.6\textwidth}
        \includegraphics[width=\linewidth]{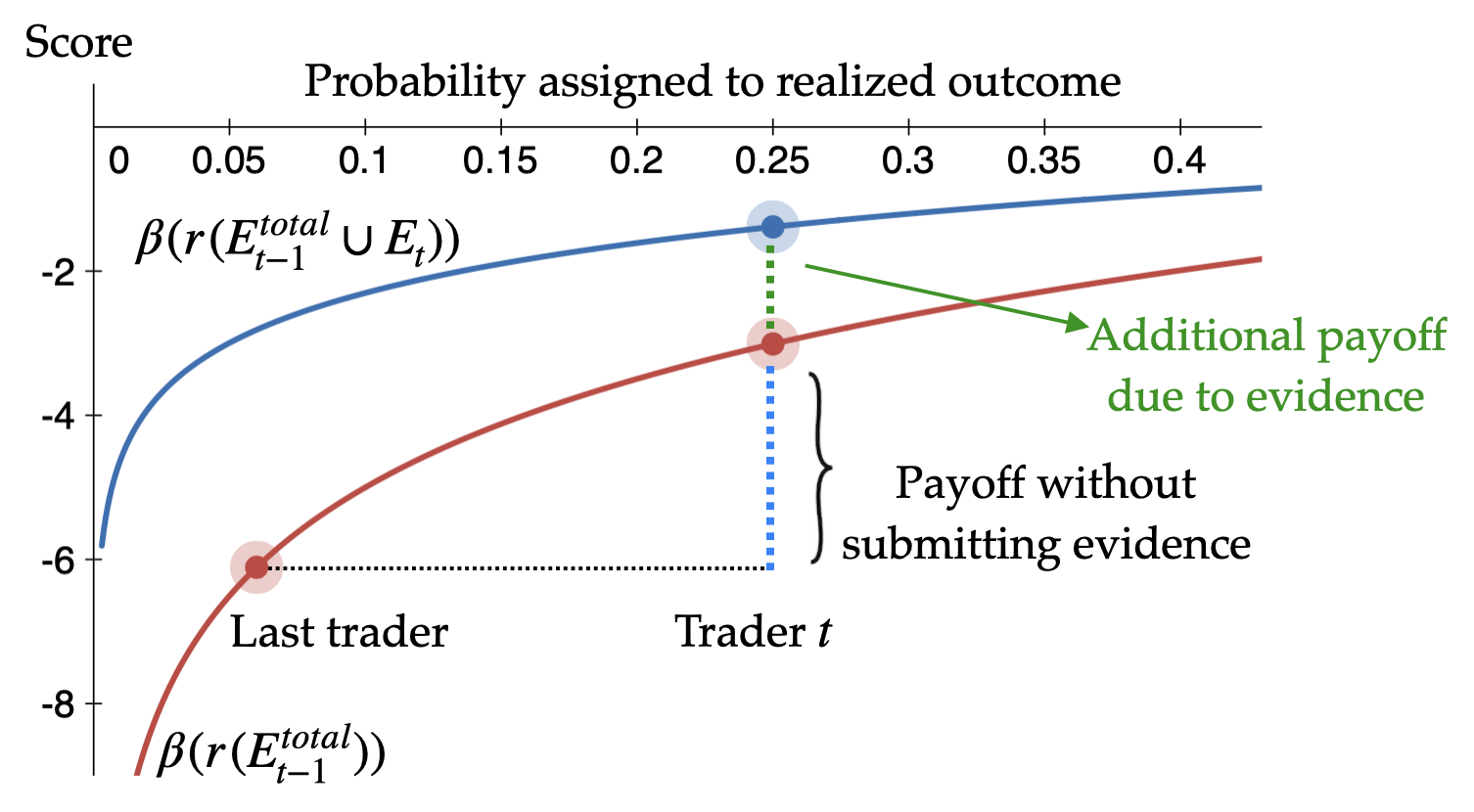}
    \end{minipage}\hfill
    \begin{minipage}[c]{0.36\textwidth}
        \caption{Log Scoring curves under different liquidity parameters. The red denotes the curve before evidence submission, and in blue, the curve after evidence submission. We visualize trader $t-1$ submitting probability 0.05 on realized outcome and trader $t$ submitting probability 0.25 along with evidence $E_t$.}
        \label{fig:evidence_payoff}
    \end{minipage}
\end{figure}

We now focus on the last desiderata and argue that traders are incentivized to submit their complete evidence set. This depends crucially on the shape of this liquidity curve $\beta(\cdot)$. Since log probabilities are negative, for positive $\beta(R)$ that \emph{decreases} in $R$, the log probability curve shifts up and becomes closer to 0. This means that keeping the belief constant (trader $t$ submits the same belief as $t-1$), by simply submitting quality evidence, they decrease $\beta$, resulting in non-zero payoff to this trader. When they submit a distinct belief alongside this, they obtain the standard LMSR payoff (which could be positive or negative) alongside this non-zero evidence payoff (see Figure~\ref{fig:evidence_payoff} for a visual). This is the core feature that incentivizes full evidence submission and we formalize it Theorem~\ref{thrm:evidence_ic}. For exogenous resolution, the incentive to submit complete evidence hold strictly; for endogenous resolution, we appeal to Theorem~\ref{theorem:belief_sensitivity} to give an $\varepsilon$ incentive compatibility guarantee. 

\begin{theorem}\label{thrm:evidence_ic}
    Let $\beta(R)$ be positive and decreasing in $R$ and suppose trader $t$ submits evidence set $E'_t \subseteq E_t$. Then $x = \beta(r(E_{t-1}^{total})) - \beta(r(E^{total}_{t-1} \cup E'_t))$ is nonnegative and the expected payoff to trader $t$ can be decomposed as:
    \begin{equation*}
          \text{payoff}_t(E'_t) = \underbrace{\beta(R_{t-1})\text{D}_{\text{KL}}(q_t(E'_t) || q_{t-1})}_{\text{(1) Belief payoff}} \quad + \underbrace{xH(q_t(E'_t))}_{\text{(2) Evidence payoff}}
    \end{equation*}
    Under exogenous resolution, it is strictly optimal to submit the complete evidence set $E_t$ and true exogenous belief $q_t$. Under endogenous resolution, for any $\varepsilon$, there exists a temperature $\tau$ such that they cannot increase payoff by more than $\varepsilon$ by submitting any partial set $E'_t \subseteq E_t$.
\end{theorem}
\begin{proof}
  Under exogenous resolution, $q_t$ is independent from this and under endogenous resolution, let $q_t = q_t(E')$ for brevity. Due to Proposition~\ref{prop:lmsr_belief_ic}, the trader will always submit this belief, so the payoff is now only a function of the evidence. Since $\beta$ is non-negative and decreasing and $r(E)$ is always monotone increasing, there must exist a non-negative $x$ such that:
    \begin{equation}
        \beta(R_t) = \beta(r(E^{total}_{t-1} \cup E'_t)) = \beta(r(E_{t-1}^{total})) -x = \beta(R_{t-1})-x
    \end{equation}
    Then their expected payoff can be expressed as follows, with $H(\cdot)$ denoting the Shannon Entropy function and $\text{D}_{\text{KL}}$ the KL divergence:
    \begin{align*}
        \text{payoff}_t(E'_t) &= \sum_{\omega}{q_t^{\omega}\left[ \beta(R_t) \log q_t^\omega - \beta(R_{t-1})\log q_{t-1}^{\omega}\right]} \\
        &= \beta(R_{t-1})\sum_{\omega}{q_t^{\omega} [\log q_{t}^{\omega}} - \log q_{t-1}^{\omega}] - x \sum_{\omega}{q_t(\omega)\log q_{t}^{\omega}} \\
        & = \beta(R_{t-1})\sum_{\omega}{q_t^{\omega} [\log q_{t}^{\omega}} - \log q_{t-1}^{\omega}] + x H(q_t)\\
        & = \beta(R_{t-1}) \sum_{\omega}q_t^{\omega} \log \frac{q_t^{\omega}}{q_{t-1}^{\omega}} + xH(q_t) = \underbrace{\beta(R_{t-1})\text{D}_{\text{KL}}(q_t || q_{t-1})}_{\text{(1) Belief payoff}} + \underbrace{xH(q_t)}_{\text{(2) Evidence payoff}}
    \end{align*}

    For exogenous resolution, $q_t$ has no relationship with $E'_t$. This means that the trader can only optimize the second term. However, since Shannon entropy is always non-negative, they maximize this by maximizing $x$, which is maximized by submitting the complete evidence set $E_t$ due to the monotone increasing nature of $r(E)$ and decreasing nature of $\beta(\cdot)$\footnote{If the trader has a perfectly deterministic belief -- i.e. $q_t$ is a one-hot vector -- then the entropy is 0 and they are indifferent to submitting evidence. In practice, however, it is standard to assume trader belief lie in the simplex interior since deterministic beliefs can have unbounded worst-case loss for traders and thus becomes impossible for platform to collect. In the AMM equivalence, pushing market prices to a one-hot value requires buying an unbounded number of shares.}. In other words, it is optimal to submit the complete evidence set under exogenous resolution.

    For endogenous resolution, the resolution belief that traders hold depends on the evidence they submit. Recall the resolution rule here computes a softmax over the scores for each alternative $j$, where score $\theta_j \in [0,1]$. This means that the resolution probability given a complete set of $k$ evidence items $E$ must satisfy:
    \begin{equation*}
        p^{\omega}(E) \geq \frac{1}{\sum_{\ell}\exp(\frac{1}{\tau}\theta_\ell)} \geq \frac{1}{n\exp(\tfrac{1}{\tau})} = \frac{1}{n}\exp(-\tfrac{1}{\tau}) := p_{min}
    \end{equation*}
    Since Proposition~\ref{prop:lmsr_belief_ic} ensures only true beliefs are submitted, and any true belief $q_{t-1}$ is a mixture over such beliefs (see Lemma~\ref{lemma:belief_endogenous}), such beliefs must lie in the simplex interior. In other words: $q_{t-1}^\omega \ge p_{\min} > 0$ for all $\omega$. Next, let $q_{1}, q_{2}$ denote beliefs that can be induced by trader $t$ due to evidence submission $E_1 \subseteq E_t$ and any subset $E_2\subseteq E_t$ respectively. Then we observe:
    \begin{equation}
        \text{payoff}_t(E_{1}) - \text{payoff}_t(E_{2}) = \beta(R_{t-1})\,\big[\mathrm{KL}(q_{1}\|q_{t-1}) - \mathrm{KL}(q_{2}\|q_{t-1})\big] \;+\; \big[x_1 H(q_{1}) - x_2 H(q_{2})\big]. \label{eq:dev-gain}
    \end{equation}

    Let $f(p) := \mathrm{KL}(p \,\|\, q_{t-1}) = \sum_\omega p^\omega \log p^\omega - \sum_\omega p^\omega \log q_{t-1}^\omega$ denote the KL function on the relative interior of the simplex. Observe that the gradient is given by: $\frac{\partial f}{\partial p^\omega} \;=\; \log p^\omega + 1 - \log q_{t-1}^\omega$. Let $C := \log(1/p_{\min})$. Then:
    \begin{equation*}
        \|\nabla f(\xi)\|_\infty \;\le\; 1 + |\log(p^\omega)| + |\log(q_{t-1}^\omega)| \leq 2C + 1 :=\; L_{\mathrm{KL}}.
    \end{equation*}
    By the mean value theorem, there exists $\xi$ on the segment between $q_1$ and $q_{2}$ such that $f(q_1) - f(q_2) = \nabla f(\xi)\cdot(q_1-q_{2})$. Combining this with Holder's inequality, we have (where $\delta = \|q_1-q_{2}\|_1$):
    \begin{equation*}
        \bigl|\mathrm{KL}(q_t\,\|\,q_{t-1}) - \mathrm{KL}(q_{t'}\,\|\,q_{t-1})\bigr| \;\le\; \|\nabla f(\xi)\|_\infty \cdot \|q_t-q_{t'}\|_1 \;\le\; L_{\mathrm{KL}}\,\delta.
    \end{equation*}

    Now for the entropy term. Let $E_1 = E_t$ and $E_2$ any subset as before. Then due to $r$ being monotone increasing and $\beta$ decreasing and non-negative, it must be that $x_1 \geq x_2$. We can thus decompose:
    \begin{equation*}
        x_1H(q_1) - x_2H(q_{2}) \;=\; (x_1-x_2)\,H(q_1) + x_2\,\big[H(q_1)-H(q_2)\big]
    \end{equation*}
    The first term is always nonnegative and thus favours full evidence submission. So we focus on the second term. Observe that for Shannon entropy $H(p)$, $\partial H/\partial p^\omega = -\log p^\omega - 1$. Over the line segment defined by $q_1, q_{2}$, the norm of this gradient is bounded by $\log(1/p_{\min}) + 1 =: L_H$. Thus by applying the mean value theorem and noting that $x_2 \leq \beta(R_{t-1}) \leq \beta_0$, we have:
    \begin{equation*}
        x_1H(q_1) - x_2H(q_{2}) \geq -\beta_0 L_{H} \delta
    \end{equation*}
    Combining it all together, we have that the payoff from submitting the complete evidence set $E_1 = E_t$ is at most $\delta(L_{KL}\beta_0 + L_{H}\beta_0)$ worse than submitting any other set $E_2$. From Theorem~\ref{theorem:belief_sensitivity}, we know that the belief sensitivity $||q_1 - q_{2}||_1 = \delta \leq \frac{|E_t|}{\tau K}$. Thus:
    \begin{equation*}
        \text{payoff}_t(E_{t}) - \text{payoff}_t(E'_{t}) \geq -\frac{\beta_0 (L_{KL} + L_{H}) |E_t|}{\tau K} := -\varepsilon  \implies \tau \geq \frac{\beta_0 (L_{KL} + L_{H})|E_t|}{\varepsilon K}
    \end{equation*}
\end{proof}

Given the result above, our market satisfies all axioms set out in Definition~\ref{def:axioms}. It can be resolved both exogenously and endogenously (Axiom 1) and incentivizes traders to submit their complete evidence set (Axiom 2) and corresponding true belief (Axiom 3) --- exactly for the exogenous case and up to some arbitrary $\varepsilon$ in the endogenous case. The market belief in an LMSR is always the last belief submitted (Axiom 4) and it is clear from the decomposition above that jointly submitting evidence and belief yields the same payoff as splitting this into two consecutive trades in either order, which ameliorates any arbitrage concern (Axiom 5). This also highlights that risk-averse traders in this market can achieve a \emph{risk-free} payout through evidence submission. Specifically, a risk-averse trader $t$ can submit $q_t = q_{t-1}$ and their evidence set $E_t$ to attain a payoff of $[\beta(R_{t-1}) - \beta(R_t)]H(q_{t-1}) \geq 0$. Since the last trader's belief captures the market belief, the marginal value of evidence is highest when the market is most entropic and the liquidity curve is the steepest. This is an attractive property that goes beyond the core axioms highlighted in Section~\ref{section:model_evidence}. We state this formally below:

\begin{corollary}
    Consider a risk-averse trader $t$ who will not engage with a market if any negative payoff is possible. By submitting $q_t = q_{t-1}$ and their evidence set $E_t$, they achieve a non-negative payoff in every realization with expected value $[\beta(R_{t-1}) - \beta(R_t)]H(q_{t-1}) \geq 0$. 
\end{corollary}

\subsection{An Equivalent Automated Market-Maker}\label{sec:amm_payoff}
While the evidence augmented LMSR satisfies all our core desiderata, it suffers from two operational challenges. In LMSRs, traders submit beliefs without any stake/money; it is only upon resolution that a platform must collect any amount they are owed. This can lead to collection issues in practice. Second, traders in practice are more accustomed to equity markets which allow them to buy "shares" of different outcomes and hold for payoffs. This is also the design of markets like Kalshi and Polymarket where traders can purchase shares of different alternatives. 

Fortunately, LMSRs can be equivalently implemented through automated market makers (AMMs) that offer such an interaction and alleviate the twin concerns above. Each ``share'' of an alternative is an Arrow-Debreu contract that pays out $\$1$ if that alternative is realized, and $\$0$ if not. The core argument rests on introducing a \emph{cost-function} that takes as input the total number of shares sold for each alternative so far. The derivative of this denotes the instantaneous share price of any alternative. A trader's execution price for a batch of shares is determined by integrating over their orders. The classic result of \cite{hanson2007logarithmic, chen2007utility} show that incentives in this market coincide exactly with LMSRs.

We now show this equivalence can be extended to our evidence-augmented LMSR. That is, we propose a modified cost function defined over both cumulative shares and evidence, whose derivative defines the marginal price of shares and evidence. Formally, let $s_{i,t}$ denote the number of shares (Arrow-Debreu contracts) that have been sold for alternative $i$ up until time $t$, with $\bs_t \in \mathbb{R}_{\geq 0}^{n}$ denoting the vector of shares sold\footnote{We do not consider/allow traders to be short in the AMM. From an information aggregation perspective where trading shares is simply a conduit to expressing beliefs, any belief that can be expressed through shorting shares, can be equivalently expressed through purchasing shares due to the zero-sum nature of the prices. For binary alternatives for instance, shorting alternative $A$ is equivalent to buying alternative $B$}. Then we define the cost-function as follows:

\begin{definition}[Evidence-Augmented AMM]\label{defn:cost_function_amm}
    Consider a market at time $t$ where $\bs_t$ shares have been sold and cumulative evidence of $R_t$ collected. Then we define the evidence-augmented cost function $C(\bs_t, R_t)$ and price function $\pi(\bs_t, R_t)$ at this state as:
    \begin{equation}
        C(\bs_t; R_t) = \beta(R_t)\log \sum_{i}{\exp\left(\frac{s_{i,t}}{\beta(R_t)}\right)} \quad ; \quad \pi_i(\bs_t; R_t) = \frac{\partial C}{\partial s_{i,t}} = \frac{\exp(\tfrac{s_{i,t}}{\beta(R_t)})}{\sum_{i}{\exp(\tfrac{s_{i,t}}{\beta(R_t)})}}
    \end{equation}
    The partial derivative of cost, $\pi_i(\bs_t; R_t)$, is the instantaneous price of buying 1 share of alternative $i$, and $\frac{\partial C}{\partial R}$ can be seen as the price of providing evidence. A trader looking to buy (sell) $\bz_t \in \mathbb{R}^{n}$ shares and supply evidence $E_t$ then pays (receives) \footnote{Recall $R_{t-1} = r(E_{t-1}^{total})$.}:
    \begin{equation}
        C(\bs_{t-1} + \bz_t, r(E_{t-1}^{total} \cup E_t)) - C(\bs_{t-1}, R_{t-1}) = \int_{(\bs_{t-1}, R_{t-1})}^{(\bs_{t-1} + \bz_t, R_t)}{\frac{\partial C}{\partial R}dR + \sum_{i=1}^{n}{\pi_i(\bs; R)d \bs}}
    \end{equation}
\end{definition}

Observe that the prices here are not ``locked in'' for a given order. Every increment of the order (shares or evidence) being executed changes the prices at which the next increment is executed. That said, since the prices are the gradient of a potential function (i.e. the cost function), the corresponding vector field is conservative and thus, any integral is path independent. This means the manner in which the order is executed does not matter, giving us the no-arbitrage behaviour we desire (Axiom 5). It is also immediate from the cost function that the sum of the prices for the $n$ alternatives sum to 1. They can thus be interpreted as the current market belief over the outcomes (Axiom 4).

The core novelty in our cost expression as compared to the classic one proposed by \citet{hanson2007logarithmic, chen2007utility} is the dependence on the evidence quality $R$ in the liquidity parameter $\beta(\cdot)$. Indeed, this dynamism has a natural interpretation in the AMM setting. Rational and risk-neutral traders have some true belief $q_t$ over outcomes which can be distinct from the current market prices/belief. As their order is executed, prices shift, and if they are not budget-constrained, they trade until the prices reflect their belief. By supplying evidence and decreasing $\beta$, they make the price curve $\pi$ steeper, implying less shares need to be bought to shift the market belief to their $q_t$. We now formalize this intuition and more generally show that for such traders their optimal payoff under this AMM market coincides exactly with the evidence-augmented LMSR, thereby inheriting all its strong incentive properties. 

\begin{theorem}\label{thm:lmsr_amm_equivalence}
    For an evidence augmented AMM at state $(\bs_{t-1}, R_{t-1})$, the expected payoff from a trade $(\bz_t, E_t)$ is equivalent to the expected profit from updating an evidence-augmented LMSR from belief $q_{t-1} = \pi(\bs_{t-1}; R_{t-1})$ and liquidity $\beta(R_{t-1})$ to a belief $q_t = \pi(\bs_{t-1} + \bz_t; r(E^{total}_{t-1} \cup E_t))$ and liquidity $\beta(r(E_{t-1}^{total} \cup E_t))$. Then for a trader who has true belief $q_t$ and evidence $E_t$, there exists a trade $\bz_t$ such that $\pi(\bs_{t-1} + \bz_t; r(E^{total}_{t-1} \cup E_t))= q_t$ and it is $\varepsilon$ incentive compatible for the trader to submit trade $(\bz_t, E_t)$ as opposed to any other trade.
\end{theorem}
\begin{proof}
    Suppose the realized outcome at resolution is alternative $i$. Let us look at the payoff under the AMM market. The trader earns $\bz_{i,t}$ for the winning alternative but pays the cost of transaction:
    \begin{align*}
        \text{AMM payoff} &= \bz_{i,t} - [C(\bs_{t-1} + \bz_t, r(E^{total}_{t-1} \cup E_t)) - C(\bs_{t-1}, R_{t-1})] \\
        &= \bz_{i,t} - \left[ \beta(r(E^{total}_{t-1} \cup E_t))\log \sum_{j} \exp\left(\frac{s_{j, t-1} + z_{j,t}}{\beta(r(E^{total}_{t-1} \cup E_t))}\right) - \beta(R_{t-1})\log \sum_{j} \exp\left(\frac{\bs_{j, t-1}}{\beta(R_{t-1})}\right) \right]
    \end{align*}

    Now consider the evidence-augmented LMSR where the last reported belief is $q_{t-1} = \pi(\bs_{t-1}; R_{t-1})$. Consider the true belief of an trader to be $q_t = \pi(\bs_{t-1} + \bz_t; r(E^{total}_{t-1} \cup E_t))$ and they have evidence $E_t$. Since the modified LMSR is incentive compatible, the trader will report this true value and submit their evidence. Their payoff then is:
    \begin{align*}
        \text{LMSR payoff} &= \beta(r(E^{total}_{t-1} \cup E_t))\log \pi(\bs_{t-1} + \bz_t; r(E^{total}_{t-1} \cup E_t)) - \beta(R_{t-1})\log \pi(\bs_{t-1}; R_{t-1}) \\
        &= \beta(r(E^{total}_{t-1} \cup E_t)) \log \left[ \frac{\exp(\tfrac{s_{i, t-1} + \bz_{i,t}}{\beta(r(E^{total}_{t-1} \cup E_t))})}{\sum_{j}{\exp(\tfrac{s_{j, t-1} + \bz_{j,t}}{\beta(r(E^{total}_{t-1} \cup E_t))})}} \right] - \beta(R_{t-1})\log\left[ \frac{\exp(\tfrac{s_{i, t-1}}{\beta(R_{t-1})})}{\sum_{j}{\exp(\tfrac{s_{j,t-1}}{\beta(R_{t-1})})}}\right] \\
        &= (\bs_{i, t-1} + \bz_{i,t}) - \beta(r(E^{total}_{t-1} \cup E_t))\log \sum_{j}{\exp(\tfrac{s_{j,t-1} + \bz_{j,t}}{\beta(r(E^{total}_{t-1} \cup E_t))})} - \bs_{i, t-1} + \beta(R_{t-1}) \log\sum_{j}{\exp(\tfrac{s_{j, t-1}}{\beta(R_{t-1})})} \\
        &= \bz_{i, t} - \left[ \beta(r(E^{total}_{t-1} \cup E_t))\log \sum_{j}{\exp(\tfrac{s_{j, t-1} + z_{j, t}}{\beta(r(E^{total}_{t-1} \cup E_t))})} - \beta(R_{t-1})\log \sum_{j}{\exp(\tfrac{s_{j, t-1}}{\beta(R_{t-1})})} \right]
    \end{align*}
    which we observe to be the same payoff under the AMM. Next note that in the evidence-augmented LMSR, submitting their true belief $q_t$ along with the evidence $E_t$ is $\varepsilon$ optimal for the trader. Next we note that since the price function is essentially a scaled softmax function, it is surjective over the simplex. So for any belief $q'$ and evidence $E'_t \subseteq E_t$, there exists a trade $\bz'$ such that $\pi(\bs_{t-1} + \bz'; r(E^{total}_{t-1} \cup E'_t)) = q'$. Now suppose that the trader chooses a $(\bz', e')$ such that $q' \ne q$. As we show above, this would lead to the same payoff as submitting belief $q' \ne q$ under the modified LMSR, which we know leads to strictly worse payoff than submitting true beliefs (Proposition~\ref{prop:lmsr_belief_ic}). Next, consider an trader choosing a $(\bz', E_t' \subset E_t)$ such that $\pi(\bs_{t-1} + \bz'; r(E^{total}_{t-1} \cup E'_t) = q_t$. Again, as we show above, this would lead to the same payoff as submitting belief $q_t$ with partial evidence $E_t' \subset E_t$ under modified LMSR. But we know from Theorem~\ref{thrm:evidence_ic} that the trader payoff in the LMSR will be $-\varepsilon$ higher in submitting the full evidence along with the true belief. Thus, the trader only submits the trade with full evidence that updates market price to $q$.
\end{proof}
With the result above, we satisfy the incentive compatibility desiderata of Axioms 2 and 3. Further, since any trader interaction here has an analogue in the LSMR realm, we also inherit the $\beta_0 \log n$ bounded worst-case platform loss for running this mechanism. We lastly turn to understanding the value of evidence in this market. In the evidence-augmented LMSR, traders are able to directly submit belief and evidence. As such, we decomposed the payoff into the value of belief update and the value of evidence submission. The latter was shown to be proportional to the entropy of the market belief. 

In the evidence-augmented AMM, there is no explicit way to submit "belief" since beliefs are endogenously defined as a function of both shares bought and evidence submitted. Even if one contributes evidence, the market belief shifts since the price curve re-adjusts to be steeper. So while the equilibrium price after a rational risk-neutral trader completes their order (buying shares and submitting evidence) $\pi(\bs_{t-1} + \bz, r(E^{total}_{t-1} \cup E_t))$ should match their true belief $q_t$ there is no clean decomposition of the AMM payoff into belief update and evidence update in the style done in LMSR. That said, a more natural decomposition in the AMM setting is on the execution cost based on the shares bought $\bz$ and evidence submitted $E_t$:
\begin{gather*}
    \text{Execution Cost: }C(\bs_{t-1} + \bz_t, r(E^{total}_{t-1} \cup E_t)) - C(\bs_{t-1}, R_{t-1})\\
    = \underbrace{C(\bs_{t-1}, r(E^{total}_{t-1} \cup E_t)) - C(\bs_{t-1}, R_{t-1})}_{\text{cost of evidence}} + \underbrace{C(\bs_{t-1}+ \bz_t, r(E^{total}_{t-1} \cup E_t)) - C(\bs_{t-1}, r(E^{total}_{t-1} \cup E_t))}_{\text{cost of position/shares}}
\end{gather*}
If the cost of the first term is shown to be negative (i.e the trader received a discount for submitting evidence) then we spiritually recover the same incentives for risk averse traders: if one wishes to take on no risk (or is budget constrained), they can buy no shares and can still get full benefit for submitting evidence. It is clear from Definition~\ref{defn:cost_function_amm} that this evidence cost is exactly equal to $\int_{R_{t-1}}^{R_t}{\frac{\partial C}{\partial R}}$ where $\frac{\partial C}{\partial R}$ is the marginal cost of evidence submission. We now show that for decreasing $\beta$, this is negative and proportional to the entropy of the market prices, spiritually similar to the LMSR evidence payoff result.
\begin{theorem}\label{thrm:marginal_evidence_cost}
    At any market state $(\bs, R)$, the marginal cost of evidence is $\frac{\partial C}{\partial R} = \frac{\partial \beta}{\partial R}H(\pi(\bs, R))$. As such, $C(\bs_{t-1}, r(E^{total}_{t-1} \cup E_t)) - C(\bs_{t-1}, R_{t-1})$ is always non-positive for any decreasing $\beta(R)$ function.
\end{theorem}
\begin{proof}
    Let $K = \sum_{i}{\exp\left(\frac{s_i}{\beta(R)} \right)}$. The price of information/evidence $\frac{\partial C}{\partial R}$ is given as follows:
    \begin{align*}
        \frac{\partial C}{\partial R} &= \frac{\partial \beta}{\partial R} \log(K) + \frac{\beta(R)}{K}\frac{\partial K}{\partial R} \quad ; \quad  \frac{\partial K}{\partial R} = - \sum_{i}{\exp\left(\frac{s_i}{\beta(R)}\right)\frac{s_i}{\beta^2(R)}\frac{\partial \beta}{\partial R}}
    \end{align*}
    We therefore have that:
    \begin{align*}
        \frac{\partial C}{\partial R} & = \frac{\partial \beta}{\partial R}\left[ \log(K) - \frac{1}{K}\sum_{i}{\exp\left(\frac{s_i}{\beta(R)}\right)\frac{s_i}{\beta(R)}} \right] = \frac{\partial \beta}{\partial R}\left[ \log(K) - \sum_{i}{\underbrace{\frac{1}{K}\exp\left(\frac{s_i}{\beta(R)}\right)}_{\pi_i(\bs; R)}\frac{s_i}{\beta(R)}} \right]
    \end{align*}
    Next, observe that by definition $\log(\pi_i) = \frac{s_i}{\beta(R)} - \log(K)$. Thus, by substituting in this relationship to the expression above, we have:
    \begin{align*}
        \frac{\partial C}{\partial R} &= \frac{\partial \beta}{\partial R}\left[ \log(K) - \sum_{i}{\pi_i(\bs; R)(\log \pi_i(\bs; R) + \log(K))}\right] \\
        &= \frac{\partial \beta}{\partial R}\left[ \log(K) + H(\pi(\bs; R)) - \log(K)\sum_{i}{\pi_i(\bs; R)}\right] = \frac{\partial \beta}{\partial R}H(\pi(\bs; R))
    \end{align*}
    where $H(p) = -\sum_{i}{p_i \log p_i}$ is the Shannon entropy. Since entropy is always positive, $\frac{\partial C}{\partial R}$ has the same sign as $\frac{\partial \beta}{\partial R}$, which is decreasing by definition.
\end{proof}
Beneath the equivalence of the two market forms lies a subtle asymmetry. For a risk-neutral trader the two representations coincide exactly, just as the classical correspondence between scoring rules and market makers predicts. The risk-averse, evidence-only trade is where they diverge. In the scoring-rule view, such a trader reports an unchanged belief, $q_t = q_{t-1}$, along with evidence $E_t$; the market belief is by definition left untouched; only the liquidity, and hence the payoff, responds to the new evidence. The corresponding risk-avers trade AMM trade is to buy nothing -- $\bz = 0$ and submit evidence $E_t$. This does, however, move the market belief since $\pi(\bs; R)$ depends on the evidence record $R$ as well as on the shares $\bs$. So fresh evidence reprices the book even when no shares change hands. The same evidence is therefore belief-neutral in one representation but belief-moving in the other. In short, evidence-only trades differ in both their payoff and their effect on the market belief between the modified LMSR and AMM implementations of our market\footnote{In both cases, however, the evidence only payoff is proportional to the current entropy of beliefs.}.


\section{Practical Considerations}\label{section:practical}
\subsection{Evidence Verification}\label{sec:verification}
Our work so far considers one axis of strategic reasoning about evidence: withholding. However, traders can also (knowingly or not) submit dubious evidence. Verification is thus a first-class concern for any practical instantiation of an evidence market. Indeed, the entire incentive structure derived in Sections~\ref{sec:endogenous_resolution} and $\ref{section:payoff}$ rests on the assumption that the quality function $r(\cdot)$ faithfully reflects the informativeness of submitted evidence. If a trader can affect endogenous resolution or drive up $R_t$ and harvest the entropy payoff by submitting irrelevant, fabricated, or duplicate content, then the platform loses out in both money and poor information aggregation. The health of the market thus rests on being able to filter out such submissions. The quality function $r$ introduced in Definition~\ref{def:evidence_quality} is where we model this filtering taking place. Conceptually, $r$ should assign zero marginal contribution to evidence that is (i) \emph{irrelevant} to the event or the alternatives being considered, (ii) \emph{wrong} or factually/objectively incorrect, or (iii) \emph{a semantic duplicate} of evidence already collected in $E^{total}_{t-1}$. The monotonicity assumption $r(E_1) \leq r(E_2)$ for $E_1 \subseteq E_2$ then ensures that such evidence has no effect on the market\footnote{The current model does not penalize such dubious submission - it just gives zero benefit for doing so. In practice, it may be fruitful to consider additional external incentives to discourage such behavior. Staking to submit evidence may be one such approach.}.

In any case, the choice of $r$ depends significantly on whether resolution is exogenous or endogenous. In the exogenous case, evidence only affects payoff magnitudes and explanation informativeness; the resolution itself is determined externally. Verification thus needs to filter out content that would waste platform funds or pollute the aggregation, but the consequences of a verification error are limited to just this. In the endogenous case, by contrast, evidence directly determines resolution, and a trader injecting fabricated evidence can shift resolution probabilities in their favor, thereby increasing payoff in both the evidence and belief axes. The verification machinery must accordingly be more robust.

A first instinct in both cases may be to turn to the rich literature on peer prediction~\citep{miller2005eliciting}. These are mechanisms that truthfully elicit opinions from agents about unverifiable signals (absence of ground truth) by rewarding them on how their reports relate to those of their peers. Unfortunately, the incentive guarantees here typically require peers to have no external incentives, which is generally not possible in our setting since a peer could also be a trader. If they have an exiting trade in the market and resolution is endogenous, they are incentivized to only approve signals that confirm their trade. If they may trade in the future, they are incentivized to not approve any evidence (in either endogenous or exogenous resolution) since it would eat up liquidity (i.e. decrease $\frac{\partial \beta}{\partial R}$) that they would otherwise access. \citet{freeman2017crowdsourced} show that ensuring incentive compatibility with general external factors essentially requires increasing verification payments and transaction costs. Given the external incentive in the endogenous case is prediction market payoff, this could be quite large and this approach would be generally unpalatable. As such, we consider verification through the LLM-as-a-Judge paradigm, with different implementations for the endogenous and exogenous resolution cases. In both cases, the LLM inference cost for verification is expected to come from transaction fees.  

\paragraph{Endogenous Resolution.} 
We start by considering what an appropriate quality function $r(\cdot)$ looks like in this regime by using the LLM-evaluation example. Recall that in the endogenous setting each atomic evidence $e$ has a binary outcome $X_{ie} \in \{0,1\}$ with respect to each alternative $i$; for LLM evaluation $X_{ie} = 1$ corresponds to model $i$ answering question $e$ correctly. One natural per-evidence quality metric is the \emph{discriminative score}: the ratio of alternatives that fail on $e$ to those that succeed on $e$. Evidence on which no alternative succeeds receives zero credit.

\begin{definition}[Discriminative Score]\label{def:discriminative_score}
    For an atomic evidence $e$ with binary outcomes $\{X_{ie}\}_{i=1}^n$, let $c(e) := \sum_i X_{ie}$ count the alternatives supporting $e$ and $w(e) := n - c(e)$ count those refuting it. The per-evidence discriminative score is
    \begin{equation}\label{eq:discriminative_score}
        r_{\mathrm{disc}}(e) \;:=\; \begin{cases} w(e) / c(e), & c(e) \geq 1, \\ 0, & c(e) = 0, \end{cases} \quad \text{and} \quad r(E) := \sum_{e \in E} r_{\mathrm{disc}}(e)
    \end{equation}
\end{definition}

$r_{\text{disc}}$ is non-negative and (additively) monotone by construction, satisfying Definition~\ref{def:evidence_quality}. The per-evidence score has a clean interpretation; it is maximized precisely when exactly one alternative is corroborated by $e$ and all others are negated. This nicely aligns with the ideal discriminating question in an LLM-evaluation context where the goal is to find questions where a few models answer correct and others answer wrong. This score drops to zero when the evidence is non-discriminating -- either all alternatives are corroborated or negated by $e$. 

Given this concrete quality function, the verification machinery's task reduces to a single binary decision per submission: should a submitted piece of evidence $e$ be admitted to $E^{total}$, with the outcomes $\{X_{ie}\}_i$ then computed and used for resolution. As discussed, verification must reject $e$ if it is irrelevant, incorrect, or a duplicate. Given the aforementioned challenges of peer prediction, we propose an \emph{LLM verification with staked disputes} mechanism, drawing on the proof-of-stake paradigm in distributed systems~\citep{buterin2017casper} and the optimistic-rollup design pattern from Ethereum scaling~\citep{kalodner2018arbitrum}. When evidence $e$ is submitted, a frontier LLM judge --- equipped with broad world knowledge and existing market context such as $E^{total}_{t-1}$ --- issues a tentative accept/reject decision along with a rationale. A bounded dispute window then opens, during which any market participant may challenge the decision by staking capital. The right to dispute is allocated via a second-price auction over stakes, with the winning challenger required to submit a written counter-argument. The judge re-adjudicates with the full history in context (the original evidence, its initial verdict and rationale, and the disputer's counter-argument). If the re-adjudication overturns the initial decision, the disputer's stake is returned and a small reward is paid out from the platform's transaction-fee pool; otherwise the stake is slashed by the platform.

This design inherits the core security argument from proof-of-stake. Economic skin in the game replaces any assumption that participants are intrinsically honest, and slashing penalties incentivizes honest disputes. Crucially, the mechanism sidesteps the peer-prediction objection raised earlier due to the LLM judge acting as the gate-keeper. Peers are not asked to truthfully report a signal in the presence of side-channel incentives; they are instead asked to stake capital to override an LLM decision they believe to be wrong. A trader with private information that the LLM erred has positive expected value in disputing; a trader without such information faces zero expected gain and will not dispute frivolously. Staking thus provides a costly-but-available channel for traders to inject information when they deem the verifier to have erred. The frontier LLM on the other hand serves as the arbiter given multiple pieces of evidence, as opposed to a one shot oracle. 

From a practical perspective, the dispute window must be bounded so verification eventually terminates. Upper and lower bounds need to be enforced for staking to deter spurious or capital rich adversaries from dominating disputes. The inference cost of running the judge must be funded by the platform through transaction fees, with slashed dispute stakes providing additional revenue. Lastly, we cannot claim the mechanism is bulletproof since the LLM judge remains a centralized point of failure. However, we believe the increasing ability of LLMs will only minimize such errors.

\paragraph{Exogenous Resolution.} When resolution is determined by an external time-bound event, evidence plays no role in adjudicating the outcome. It only enriches the aggregated information product and scales the liquidity parameter. Verification is correspondingly less critical and we simply need to design $r$ to ensure the platform is not paying out for and aggregating poor quality information. For this we propose a lightweight scheme inspired by the LLM-judge construction of \citet{srinivasan2025tellmewhy}.

Concretely, consider a frontier LLM judge $J$ equipped with broad world knowledge and the cumulative evidence record $E^{total}_{t-1}$. We elicit from $J$ a belief distribution $p_J(\,\cdot\, ; E^{total}_{t-1})$ over the $n$ alternatives. When a new submission $E_t$ arrives, we append it to the judge's context and re-elicit a posterior belief $p_J(\,\cdot\, ; E^{total}_{t-1} \cup E_t)$. The marginal quality contribution of the new evidence is then defined as the divergence between these two distributions:
\begin{equation}\label{eq:r_exogenous}
    r(E^{total}_{t-1} \cup E_t) - r(E^{total}_{t-1}) := D_{KL}\big(p_J(\cdot \, ; E^{total}_{t-1} \cup E_t) \,\|\, p_J(\cdot \, ; E^{total}_{t-1})\big),
\end{equation}
with $R_t$ being the running sum of such contributions and $r(\emptyset) = 0$. This definition satisfies the non-negativity and monotonicity requirements of Definition~\ref{def:evidence_quality} by construction, since KL divergence is non-negative. Indeed, KL-divergence in \eqref{eq:r_exogenous} can be replaced by any other f-divergence (total variation, Jensen-Shannon) without affecting anything qualitatively.

The key property of this construction is that a competent judge should be \emph{unmoved} by evidence lacking genuine informational content. If $E_t$ is irrelevant to the event in question or is already contained in $E_{t-1}^{total}$, the judge has no reason to revise its belief; if $E_t$ is factually wrong, the judge's world knowledge should lead it to discount the submission. In these failure modes, the divergence in \eqref{eq:r_exogenous} collapses to zero (or near-zero), and the trader receives no evidence payoff. Conversely, evidence that is novel, correct, and pertinent will move the judge's belief and accrue credit proportional to its informational content. This aligns the incentives of evidence submission with the platform's epistemic goals: traders are paid for information that genuinely changes a well-informed observer's view of the event. The overall scheme inherits the verification asymmetry discussed in Section~\ref{sec:related_works} --- the judge need only \emph{evaluate} a submitted piece of evidence rather than generate it, which is structurally easier and well-supported by recent work on LLM verifiers~\citep{lightman2023verify, zhang2024genrm, khalifa2025thinkprm}. As before, while the LLM-judge does remain a central point of failure, it becomes less so as their power and capabilities improve.

\subsection{Execution Algorithms}\label{sec:execution}
The AMM mechanism developed in Section~\ref{section:payoff} implicitly assumes that trades arrive and are executed sequentially: each order's evidence is fully verified, $r(E_t)$ computed, the cumulative quality $R_t$ updated, the order executed before the next trader interacts with the market. In practice this is unworkable. Verification, particularly through an LLM judge of the kind described just earlier, takes meaningful time, and a market that pauses for each verification will simply be too slow for most traders. Concretely, consider a trader $t$ who submits an order $(\bz_t, E_t)$ and is followed seconds later by a trader $t+1$ with $E_{t+1} = \emptyset$. Verification of $E_t$ may not have completed by the time $t+1$ arrives. What price should the market quote $t+1$? Should they have to wait until $E_t$ is verified to trade? 

We address this through an asynchronous execution protocol consisting of two concurrent algorithms: a \emph{trade execution} routine that processes incoming orders immediately at a conservative price, and a \emph{verification adjustment} routine that consumes the evidence queue in arrival order and refunds the difference between the conservative charge and the true synchronous cost. The protocol is designed to satisfy four desiderata:

\begin{itemize}[itemsep=0.0ex,leftmargin=4ex]
    \item \textbf{Equivalence to synchronous execution.} Each trader's final net payment matches what they would have paid under sequential synchronous execution, preserving all incentive guarantees of Section~\ref{section:payoff}.
    \item \textbf{Immediate execution for evidence-free orders.} A trader submitting $E_t = \emptyset$ never waits on any verification queue.
    \item \textbf{At-most-one charge per order.} Traders can only be charged (transfer of money from agent to platform) at-most once during the execution of a single order. Any subsequent transfer is a rebate from platform to trader.
    \item \textbf{Lower-bounded refund at execution.} The platform can quote the trader a guaranteed minimum refund at the time the order is executed.
\end{itemize}

The central design choice is to execute each order \emph{pessimistically}, i.e.~charge the maximum cost the trader could face once the unverified evidence resolves and then refund the difference once verification completes. Whether the pessimistic regime corresponds to higher or lower $\beta$ depends on the order, since traders do not uniformly prefer high liquidity. The following lemma characterizes this dependence and gives us a computable worst-case $\beta$ within any liquidity interval.

\begin{lemma}\label{lem:beta_grid_opt}
    For fixed cumulative shares $\bs$, share order $\bz$, and liquidity $\beta$, let $\xi(\bz; \bs, \beta) := C(\bs + \bz; \beta) - C(\bs; \beta)$ denote the execution cost of the shares. Then
    \begin{equation}\label{eq:dxi_dbeta}
        \frac{\partial \xi}{\partial \beta} \;=\; H\big(\pi(\bs + \bz; \beta)\big) - H\big(\pi(\bs; \beta)\big),
    \end{equation}
    So $\xi$ is increasing in $\beta$ when the trade raises market entropy and decreasing when it lowers it. 
\end{lemma}
\begin{proof}
    Since $C(\bs; \beta) = \beta \log K$ with $K = \sum_i \exp(s_i/\beta)$, the product rule gives
    \begin{equation*}
        \frac{\partial C}{\partial \beta} \;=\; \log K + \frac{\beta}{K} \cdot \frac{\partial K}{\partial \beta}, \qquad \frac{\partial K}{\partial \beta} \;=\; -\frac{1}{\beta^2}\sum_i s_i \exp\!\left(\frac{s_i}{\beta}\right).
    \end{equation*}
    Combining and recognizing $\pi_i(\bs; \beta) = \exp(s_i/\beta)/K$ yields $\partial C/\partial \beta = \log K - \sum_i \pi_i (s_i / \beta)$. Substituting $s_i/\beta = \log \pi_i + \log K$ and using $\sum_i \pi_i = 1$ simplifies this to $\partial C/\partial \beta = H(\pi(\bs; \beta))$. Equation~\eqref{eq:dxi_dbeta} follows immediately.
\end{proof}

We can now state the two algorithms. The trade execution algorithm accepts each incoming order $(\bz_t, E_t)$ at time $t$,  appends $E_t$ to the unverified evidence queue $Q$, and charges trader $t$ the worst-case share execution cost over the possible range of cumulative quality the addition of this evidence could lead to. Let $R$ denote the verified cumulative quality computed so far, and let $r_Q$ denote an upper bound on the total quality contribution of evidence currently in $Q$.\footnote{In practice, $r_Q$ can be set by the platform as a sum of per-evidence quality ceilings — e.g.\ by capping the divergence in equation~\eqref{eq:r_exogenous} at some platform-chosen value.} When trader $t$'s trade raises market entropy -- so $\partial \xi / \partial \beta > 0$ by Lemma~\ref{lem:beta_grid_opt} -- the trader prefers \emph{low} $\beta$. So the worst case for them is the \emph{highest} feasible $\beta$, achieved at $R^{\mathrm{ex}}_t = R$ (i.e. assume the queue contributes nothing). When the trade lowers market entropy, the worst case is instead $R^{\mathrm{ex}}_t = R + r_Q$ (assume the queue contributes its maximum). A market snapshot at the time of execution is stored for use by the verification adjustment routine.

\begin{algorithm}[H]
\caption{Trade Execution}\label{alg:trade_execution}
\KwIn{Verified state so far $(\bs, R)$; queue $Q$ with quality upper bound $r_Q$; arriving order stream}

\texttt{snapshots} $\gets \{\}$ \tcp*[r]{trader index $\to$ execution-time state}

\While{the order stream is not empty}{
    pull earliest order $t : (\bz_t, E_t)$ \\
    \uIf{$H\!\left(\pi(\bs + \bz_t; \beta(R))\right) - H\!\left(\pi(\bs; \beta(R))\right) \geq 0$}{
        $R^{\mathrm{ex}}_t \gets R$ \tcp*[r]{trade raises entropy; pessimistic: highest $\beta$}
    }
    \Else{
        $R^{\mathrm{ex}}_t \gets R + r_Q$ \tcp*[r]{trade lowers entropy; pessimistic: lowest $\beta$}
    }
    \texttt{snapshots}$[t] \gets \big(\bs, \bz_t, R^{\mathrm{ex}}_t, \text{copy of } Q\big)$ \\
    append $(t, E_t)$ to $Q$ \\
    charge trader $t$: $\xi^{\mathrm{ex}}_t := C\!\left(\bs + \bz_t; \beta(R^{\mathrm{ex}}_t)\right) - C\!\left(\bs; \beta(R^{\mathrm{ex}}_t)\right)$ \\
    update $\bs \gets \bs + \bz_t$
}
\end{algorithm}

The verification adjustment algorithm pulls evidence from $Q$ in arrival order, verifies it to compute $r(E_{t-1}^{total} \cup E_t)$, and refunds trader $t$ the difference between $\xi^{\mathrm{ex}}_t$ and the synchronous cost. This refund decomposes naturally into two non-negative components: a \emph{transaction refund} reflecting that $R^{\mathrm{ex}}_t$ executed the shares at the trader's worst-case $\beta$, and an \emph{evidence refund} reflecting the quality contribution $r(E_t)$ of the trader's own evidence.

\begin{algorithm}[H]
\caption{Verification Adjustment}\label{alg:verification_adjustment}
\KwIn{Queue $Q$ from Algorithm~\ref{alg:trade_execution}; snapshots dictionary \texttt{snapshots}}

\texttt{verified} $\gets \{\}$ \tcp*[r]{trader index $\to$ $r(E_t)$}

\While{$Q$ is not empty}{
    pop earliest $(t, E_t)$ from $Q$; verify and set \texttt{verified}$[t] \gets r(E_t)$ \\
    let $\big(\bs^{\mathrm{ex}}, \bz_t, R^{\mathrm{ex}}_t, Q_t\big) \gets$ \texttt{snapshots}$[t]$ \\
    $R^{\mathrm{true}}_t \gets \sum_{(t', E_{t'}) \in Q_t} \texttt{verified}[t']$ \tcp*[r]{queue entries preceding $t$ are now verified}
    \texttt{trans\_refund} $\gets \big[C(\bs^{\mathrm{ex}} + \bz_t; \beta(R^{\mathrm{ex}}_t)) - C(\bs^{\mathrm{ex}}; \beta(R^{\mathrm{ex}}_t))\big] - \big[C(\bs^{\mathrm{ex}} + \bz_t; \beta(R^{\mathrm{true}}_t)) - C(\bs^{\mathrm{ex}}; \beta(R^{\mathrm{true}}_t))\big]$ \\
    \texttt{evidence\_refund} $\gets C(\bs^{\mathrm{ex}} + \bz_t; \beta(R^{\mathrm{true}}_t)) - C(\bs^{\mathrm{ex}} + \bz_t; \beta(R^{\mathrm{true}}_t + r(E_t)))$ \\
    transfer \texttt{trans\_refund} $+$ \texttt{evidence\_refund} to trader $t$
}
\end{algorithm}

\begin{theorem}\label{thm:async_correctness}
    Under Algorithms~\ref{alg:trade_execution} and \ref{alg:verification_adjustment}, the final net payment from each trader $t$ matches the synchronous sequential execution cost. Each trader is charged once at execution; all subsequent transfers are non-negative refunds. Furthermore, the platform can quote each trader a deterministic lower bound on their total refund at the time of execution.
\end{theorem}
\begin{proof}
    Fix trader $t$ with snapshot $(\bs^{\mathrm{ex}}, \bz_t, R^{\mathrm{ex}}_t, Q_t)$, true cumulative quality $R^{\mathrm{true}}_t$ at the time of their arrival, and own evidence quality $r_t := r(E_t)$. Their synchronous cost is
    \begin{equation*}
        \xi^{\mathrm{sync}}_t \;:=\; C\!\left(\bs^{\mathrm{ex}} + \bz_t; \beta(R^{\mathrm{true}}_t + r_t)\right) - C\!\left(\bs^{\mathrm{ex}}; \beta(R^{\mathrm{true}}_t)\right).
    \end{equation*}
    Adding and subtracting $C(\bs^{\mathrm{ex}} + \bz_t; \beta(R^{\mathrm{true}}_t)) - C(\bs^{\mathrm{ex}}; \beta(R^{\mathrm{true}}_t))$, we obtain the identity
    \begin{equation}\label{eq:sync_decomp}
        \xi^{\mathrm{sync}}_t \;=\; \xi^{\mathrm{ex}}_t \;-\; \texttt{trans\_refund} \;-\; \texttt{evidence\_refund},
    \end{equation}
    where the two refund terms are as defined in Algorithm~\ref{alg:verification_adjustment}. Hence the trader's net payment after refund equals $\xi^{\mathrm{sync}}_t$.

    Both refund components are non-negative. For \texttt{trans\_refund}, Lemma~\ref{lem:beta_grid_opt} together with the pessimistic choice of $R^{\mathrm{ex}}_t$ implies that $\xi(\bs^{\mathrm{ex}}, \bz_t, \beta(R^{\mathrm{ex}}_t)) \geq \xi(\bs^{\mathrm{ex}}, \bz_t, \beta(\rho))$ for every $\rho$ in the plausible range $[R, R + r_Q]$, and in particular for $\rho = R^{\mathrm{true}}_t$. For \texttt{evidence\_refund}, Theorem~\ref{thrm:marginal_evidence_cost} establishes $\partial C / \partial R = (\partial \beta / \partial R) H(\pi) \leq 0$, so $C$ is non-increasing in $R$ and the difference is non-negative. The trader thus faces a single charge at execution followed only by transfers in their favor.

    Finally, the platform can lower bound the refund at execution time. Since $\texttt{trans\_refund} \geq 0$, the entire refund is bounded below by $\texttt{evidence\_refund}$, which in turn can be expressed as the line integral $\int_{R^{\mathrm{true}}_t}^{R^{\mathrm{true}}_t + r_t} |\partial \beta / \partial R| \cdot H(\pi(\bs^{\mathrm{ex}} + \bz_t; \beta(R'))) dR'$. Bounding the integrand below over the plausible range $R' \in [R, R + r_Q + r_t^{\max}]$ --- with $r_t^{\max}$ a platform-set ceiling on per-order quality contribution --- and assuming a platform-set minimum quality $r_t^{\min}$ for evidence to be deemed valid, gives the executable lower bound
    \begin{equation*}
        \texttt{evidence\_refund} \;\geq\; r_t^{\min} \cdot \min_{R' \in [R,\, R + r_Q + r_t^{\max}]} \left\{ \bigg|\frac{\partial \beta}{\partial R}\bigg|_{R'} \cdot H\!\left(\pi(\bs^{\mathrm{ex}} + \bz_t; \beta(R'))\right) \right\},
    \end{equation*}
    which is computable from quantities known at execution time and given to the trader as a conditional guarantee (i.e.\ contingent on their evidence passing verification with quality at least $r_t^{\min}$).
\end{proof}

\section{Discussion}\label{sec:discussion}
This work introduced evidence markets, a generalization of prediction markets in which traders can convey beliefs alongside evidence, and which resolves either exogenously through a realized event or endogenously through the submitted evidence itself. This addresses the two limitations we highlighted: the opacity of the reasoning behind a market price, and the inapplicability of prediction markets to questions that admit no time-bound external resolution. Concretely, the mechanism extends the logarithmic market scoring rule with a liquidity parameter that decreases as cumulative evidence quality grows; under it, reporting both one's belief truthfully and one's evidence wholly is strictly DSIC when the market resolves exogenously and $\varepsilon$-DSIC when it resolves endogenously. The same mechanism can equivalently be cast as an automated market maker. In both views, the trader's payoff can be cleanly split to capture the value of evidence so that even a risk-averse trader who takes no directional position earns a non-negative reward simply for contributing good evidence. Although LLM evaluation is our running example, the design fits any setting where the question is comparative and its answer must be assembled from a body of evidence rather than read off from nature; examples include the replicability of a scientific finding, the effectiveness of a policy, or the relative quality of competing artifacts.

Several questions remain. The most pressing being a deeper treatment of verification; our LLM-gated mechanism with staked disputes is a constructive starting point, but hardening it to become production-ready for an adversarial, real-money market is a substantial design and engineering problem. A second concerns the acquisition cost of evidence. We treat evidence as free to acquire, whereas producing a genuinely informative piece is often effortful, and modeling that cost and the resulting tension between reporting one's belief and actively gathering evidence for it would bring the framework closer to the reality we have in mind. A third concerns the trading interface. Platforms such as Kalshi and Polymarket run on central limit order books rather than cost-function based AMMs; how best to implement evidence markets in that setting with similar guarantees remains open. 

\section*{Acknowledgments}
We thank Sid Srinivasan, Bo Waggoner, and Raf Frongillo for their insightful comments and the many helpful discussions we have had.

\bibliographystyle{plainnat}
\bibliography{bibliography}


\end{document}